\documentclass[11pt]{article}

\usepackage[margin=1.3in]{geometry}
\usepackage[utf8]{inputenc}
\usepackage[T1]{fontenc}
\usepackage[USenglish]{babel}
\usepackage{lmodern} 
\usepackage{graphicx}
\usepackage{subcaption} \DeclareCaptionSubType[roman]{figure} 
\usepackage{enumitem} \setlist[enumerate,1]{label={(\alph*)}} 
\usepackage{microtype}
\usepackage{array}

\usepackage{tikz}
\usepackage{tkz-berge}
\usetikzlibrary{decorations.pathreplacing}
\usetikzlibrary{topaths,calc,graphs, math}
\usetikzlibrary{positioning}
\usepackage{pgfmath}
\usetikzlibrary{backgrounds}
\tikzset{
	position/.style args={#1:#2 from #3}{at=(#3), shift=(#1:#2)},
	rnode/.style={circle, draw=black, fill=white,  thin, inner sep=0pt, minimum size=4pt},
	fnode/.style={circle, draw=black, fill=black,  thin, inner sep=0pt, minimum size=8pt},
	cnode/.style={circle, draw=white, fill=white,  thin, inner sep=0pt, minimum size=0pt},
	snode/.style={circle, draw=black, fill=black,  thin, inner sep=0pt, minimum size=4pt},
	line/.style = { draw, thick, -{stealth} },
	nonarrow/.style = { draw, thick},
	dline/.style = { draw, thick, -{stealth}, dotted },
	nline/.style={draw, line width=7pt, yellow!40!white, rounded corners=.2*\pgflinewidth},
	n2line/.style={draw, line width=7pt, orange!40!white, rounded corners=.2*\pgflinewidth}
}

\definecolor{hellblau}{rgb}{0.2,0.4,1} 
\definecolor{dunkelblau}{rgb}{0,0,0.8}
\definecolor{dunkelgruen}{rgb}{0,0.5,0}
\definecolor{green}{rgb}{0,0.6,0}
\definecolor{fillblack}{rgb}{0.95,0.95,0.95}
\usepackage[
	pdftex,
	colorlinks,
	linkcolor=dunkelblau,
	urlcolor=dunkelblau,
	citecolor=dunkelgruen,
	bookmarks=true,
	linktocpage=true,
	pdfauthor={},
	pdfsubject={}
]{hyperref}

\usepackage{pdfpages} 

\usepackage{amsmath} 
\usepackage{amsthm}
\usepackage{amssymb}
\usepackage{amsfonts} 
\usepackage{mathtools} 
\theoremstyle{plain}
	\newtheorem{satz}{Satz}[] 
	\newtheorem{theorem}[satz]{Theorem}
	\newtheorem*{theorem*}{Theorem}
	\newtheorem{lemma}[satz]{Lemma}
	\newtheorem*{lemma*}{Lemma}
	\newtheorem{corollary}[satz]{Corollary}

\theoremstyle{remark}
	\newtheorem*{remark}{\textbf{Remark}} 
\theoremstyle{definition}

	\newtheorem{definition}[satz]{Definition}
	
	\newtheorem*{conjecture*}{Conjecture}

\newcommand{\red}{{\color{red}1}}
\newcommand{\green}{{\color{dunkelgruen}2}}
\newcommand{\blue}{{\textcolor{blue}3}}

\title{Toward Grünbaum's Conjecture}
\author{Christian Ortlieb\\Institute of Computer Science\\University of Rostock\\Germany\thanks{This research is supported by the grant SCHM 3186/2-1 (401348462) from the Deutsche Forschungsgemeinschaft (DFG, German Research Foundation).}
	\and Jens M. Schmidt\\Institute of Computer Science\\University of Rostock\\Germany\footnotemark[1]}
\date{}

\begin{document}
	\maketitle
	\begin{abstract}
		Given a spanning tree $T$ of a planar graph $G$, the \emph{co-tree} of $T$ is the spanning tree of the dual graph $G^*$ with edge set $(E(G)-E(T))^*$.
		Grünbaum conjectured in 1970 that every planar 3-connected graph $G$ contains a spanning tree $T$ such that both $T$ and its co-tree have maximum degree at most~3.
		
		While Grünbaum's conjecture remains open, Biedl proved that there is a spanning tree $T$ such that $T$ and its co-tree have maximum degree at most~5.
		By using new structural insights into Schnyder woods, we prove that there is a spanning tree $T$ such that $T$ and its co-tree have maximum degree at most~4.
	\end{abstract}

	\section{Introduction}
	Let a $k$-\emph{tree} be a spanning tree whose maximum degree is at most~$k$. In 1966, Barnette proved the fundamental theorem that every planar 3-connected graph contains a 3-tree~\cite{Barnette1966}. Both assumptions in this theorem are essential in the sense that the statement fails for arbitrary non-planar graphs (as the arbitrarily high degree in any spanning tree of the complete bipartite graphs $K_{3,n-3}$ show) as well as for graphs that are not 3-connected (as the planar graphs $K_{2,n-2}$ show).
	
	Since then, Barnette's theorem has been extended and generalized in several directions.
	First, one may try to relax the 3-connectedness assumption: Indeed, Barnette's original proof holds for the slightly more general class of \emph{circuit graphs}\footnote{that is, planar internally 3-connected graphs with a designated outer face}, and may also be extended to arbitrary planar graphs $G$ in form of a local version that guarantees for every 3-connected\footnote{$X \subseteq V(G)$ such that $G$ contains three internally vertex-disjoint paths between every two vertices of $X$} vertex set $X$ of $G$ a (not necessarily spanning) tree of $G$ that has maximum degree at most~3 and contains $X$~\cite{Boehme2024}. Alternatively, one may relax the planarity assumption. Ota and Ozeki~\cite{Ota2012} proved that for every $k \geq 3$, every 3-connected graph with no $K_{3,k}$-minor contains a $(k-1)$-tree if $k$ is even and a $k$-tree if $k$ is odd. Further sufficient conditions for the existence of $k$-trees may be found in the survey~\cite{Ozeki2011}.
	
	Second, one may see spanning trees as $1$-connected spanning subgraphs and generalize these to $k$-connected spanning subgraphs for any $k > 1$. In this direction, Barnette~\cite{Barnette1994} proved that every planar 3-connected planar graph contains a 2-connected spanning subgraph whose maximum degree is at most~15, and Gao~\cite{Gao1995b} improved this result subsequently to the tight bound of maximum degree at most~6. Interestingly, Gao showed that his result holds as well for the 3-connected graphs that are embeddable on the projective plane, the torus or the Klein bottle.
	
	Third, one may try to strengthen the 3-tree in question. A recent alternative proof of Barnette's theorem based on canonical orderings by Biedl~\cite[Corollary~1]{Biedl2014} (which was also mentioned by Chrobak and Kant) reveals that further degree constraints may be imposed on the 3-tree for prescribed vertices (for example, two vertices of a common face may be forced to be leaves of the tree). To strengthen this further, Barnette's theorem can be seen as a side-result of a structure obtained in Hamiltonicity studies from generalizing the theory of Tutte paths and Tutte cycles: Gao and Richter~\cite{Gao1994} proved that every planar 3-connected graph contains a $2$-\emph{walk}, which is a walk that visits every vertex exactly once or twice. By going along such 2-walks and omitting the last edge whenever a vertex is revisited, these 2-walks imply the existence of 3-trees. Here, planar 3-connected graphs may again be replaced with circuit graphs, and all results have been successfully lifted to higher surfaces. Even more, the surfaces on which every embedded 3-connected graph contains a 2-walk have been classified~\cite{Brunet1995}.
	
	Perhaps one of the most severe strengthenings of the 3-tree in question is a long-standing and to the best of our knowledge still open conjecture made by Grünbaum in~1970. Since the planar dual $G^*=(V^*,E^*)$ of every (simple) planar 3-connected graph $G$ is again planar and 3-connected, $G^*$ contains a 3-tree as well. By the well-known cut-cycle duality, any spanning tree $T$ of $G$ implies that also $\neg T^* := (V^*,(E(G)-E(T))^*)$ is a spanning tree of $G^*$; we call $\neg T^*$ the \emph{co-tree} of $T$. Taking the best of these two worlds, Grünbaum made the following conjecture.

	\begin{conjecture*}[{Grünbaum~\cite[p.~1148]{Grunbaum1970}, 1970}]
		Every planar 3-connected graph $G$ contains a 3-tree $T$ whose co-tree $\neg T^*$ is also a 3-tree.
	\end{conjecture*}
	
	While Grünbaum's conjecture is to the best of our knowledge still unsolved, progress has been made by Biedl~\cite{Biedl2014}, who proved the existence of a 5-tree, whose co-tree is a 5-tree. We prove the existence of a 4-tree, whose co-tree is a 4-tree. Our methods exploit insights into the structure of Schnyder woods. We discuss Schnyder woods, their lattice structure and ordered path partitions in Section~\ref{sec:preliminaries}, our main result in Section~\ref{sec:maxdeg4} and computational aspects of this main result in Section~\ref{sec:computational}.

	\section{Schnyder Woods and Ordered Path Partitions}\label{sec:preliminaries}
	We only consider simple undirected graphs. A graph is \emph{plane} if it is planar and embedded into the Euclidean plane. The \emph{neighborhood of a vertex set} $A$ is the union of the neighborhoods of vertices in $A$. Although parts of this paper use orientation on edges, we will always let $vw$ denote the undirected edge $\{v,w\}$.
	
	\subsection{Schnyder Woods.}
	Let $\sigma := \{r_1,r_2,r_3\}$ be a set of three vertices of the outer face boundary of a plane graph $G$ in clockwise order (but not necessarily consecutive). We call $r_1$, $r_2$ and $r_3$ \emph{roots}. The \emph{suspension} $G^\sigma$ of $G$ is the graph obtained from $G$ by adding at each root of $\sigma$ a half-edge pointing into the outer face.
	A plane graph $G$ is $\sigma$-\emph{internally 3-connected} if the graph obtained from the suspension $G^\sigma$ of $G$ by making the three half-edges incident to a common new vertex inside the outer face is 3-connected. Note that the class of $\sigma$-internally 3-connected plane graphs properly contains all 3-connected plane graphs.
	
	\begin{definition}\label{def:Schnyderwood}
		Let $\sigma = \{r_1,r_2,r_3\}$ and $G^\sigma$ be the suspension of a $\sigma$-internally 3-connected plane graph $G$. A \emph{Schnyder wood} of $G^\sigma$ is an orientation and coloring of the edges of $G^\sigma$ (including the half-edges) with the colors \red,\green,\blue\ ({\color{red}red}, {\color{dunkelgruen}green}, {\color{blue}blue}) such that
		\begin{enumerate}
			\item Every edge $e$ is oriented in one direction (we say $e$ is \emph{unidirected}) or in two opposite directions (we say $e$ is \emph{bidirected}). Every direction of an edge is colored with one of the three colors \red,\green,\blue\ (we say an edge is $i$-\emph{colored} if one of its directions has color $i$) such that the two colors $i$ and $j$ of every bidirected edge are distinct (we call such an edge $i$-$j$-\emph{colored}). Similarly, a unidirected edge whose direction has color $i$ is called $i$-\emph{colored}. Throughout the paper, we assume modular arithmetic on the colors \red,\green,\blue\ in such a way that $i+1$ and $i-1$ for a color $i$ are defined as $(i \mod 3) +1$ and $(i+1 \mod 3) +1$. For a vertex $v$, a uni- or bidirected edge is \emph{incoming} ($i$-colored) \emph{in} $v$ if it has a direction (of color $i$) that is directed toward $v$, and \emph{outgoing} ($i$-colored) \emph{of} $v$ if it has a direction (of color $i$) that is directed away from $v$.\label{def:Schnyderwood1}
			\item For every color~$i$, the half-edge at $r_i$ is unidirected, outgoing and $i$-colored.\label{def:Schnyderwood2}
			\item Every vertex $v$ has exactly one outgoing edge of every color. The outgoing \red-, \green-, \blue-colored edges $e_\red,e_\green,e_\blue$ of $v$ occur in clockwise order around $v$. For every color~$i$, every incoming $i$-colored edge of $v$ is contained in the clockwise sector around $v$ from $e_{i+1}$ to $e_{i-1}$ (see Figure~\ref{fig:SchnyderWoodCondition}).\label{def:Schnyderwood3}
			\item No inner face boundary contains a directed cycle (disregarding possible opposite edge directions) in one color.\label{def:Schnyderwood4}
		\end{enumerate}
	\end{definition}
	
	\begin{figure}[!htb]
		\centering
		\begin{tikzpicture}[scale = 0.25]
			\node[rnode] (0) at (0,0) {};
			
			\draw [line, red] (0) to (90:10);
			\draw [line, blue] (0) to (210:10);
			\draw [line, green] (0) to (330:10);
			
			\draw [line, red] (240:6) to (0);
			\draw [line, red] (270:6) to (0);
			\draw [line, red] (300:6) to (0);
			
			\draw [line, blue] (10:6) to (0);
			\draw [line, blue] (50:6) to (0);
			
			\draw [line, green] (120:6) to (0);
			\draw [line, green] (150:6) to (0);
			\draw [line, green] (180:6) to (0);
			
			\node[] (1) at (82:9) {1};
			\node[] (2) at (338:9) {2};
			\node[] (3) at (202:9) {3};
			
			\node[] (11) at (232:6) {1};
			\node[] (12) at (278:6) {1};
			\node[] (13) at (308:6) {1};
			
			\node[] (21) at (112:6) {2};
			\node[] (22) at (158:6) {2};
			\node[] (23) at (188:6) {2};
			
			\node[] (31) at (2:6) {3};
			\node[] (32) at (58:6) {3};

		\end{tikzpicture}
		\caption{Properties of Schnyder woods. Condition~\ref{def:Schnyderwood}\ref{def:Schnyderwood3} at a vertex.}
		\label{fig:SchnyderWoodCondition}
	\end{figure}
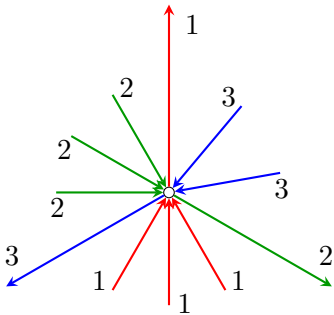
	
	For a Schnyder wood and color~$i$, let $T_i$ be the directed graph that is induced by the directed edges of color~$i$. The following result justifies the name of Schnyder woods.
	
	\begin{lemma}[\cite{Schnyder1990,Felsner2004a}]
		For every color~$i$ of a Schnyder wood of a graph $G$, $T_i$ is a directed spanning tree of $G$ in which all edges are oriented to the root $r_i$.
		\end{lemma}
		
		For a directed graph $H$ we denote by $H^{-1}$ the graph obtained from $H$ by reversing the orientation of all edges.
		
		\begin{lemma}[Felsner \cite{Felsner2001}]
			$T_i \cup T^{-1}_{i-1} \cup T^{-1}_{i+1}$ does not have any oriented cycle.
			\label{lem_felsner_no_cycle}
		\end{lemma}

		\subsection{Dual Schnyder Woods.}
		Let $G$ be a $\sigma$-internally 3-connected plane graph. Any Schnyder wood of $G^\sigma$ induces a Schnyder wood of a slightly modified planar dual of $G^\sigma$ in the following way~\cite{DiBattista1999,Felsner2004} (see~\cite[p.~30]{Kant1996} for an earlier variant of this result given without proof). As common for plane duality, we will use the plane dual operator $^*$ to switch between primal and dual objects (also on sets of objects).
		
		Extend the three half-edges of $G^\sigma$ to non-crossing infinite rays and consider the planar dual of this plane graph. Since the infinite rays partition the outer face $f$ of $G$ into three parts, this dual contains a triangle with vertices $b_1$, $b_2$ and $b_3$ instead of the outer face vertex $f^*$ such that $b^*_i$ is not incident to $r_i$ for every $i$ (see Figure~\ref{fig:completion}). Let the \emph{suspended dual} $G^{\sigma^*}$ of $G$ be the graph obtained from this dual by adding at each vertex of $\{b_1,b_2,b_3\}$ a half-edge pointing into the outer face.
		
		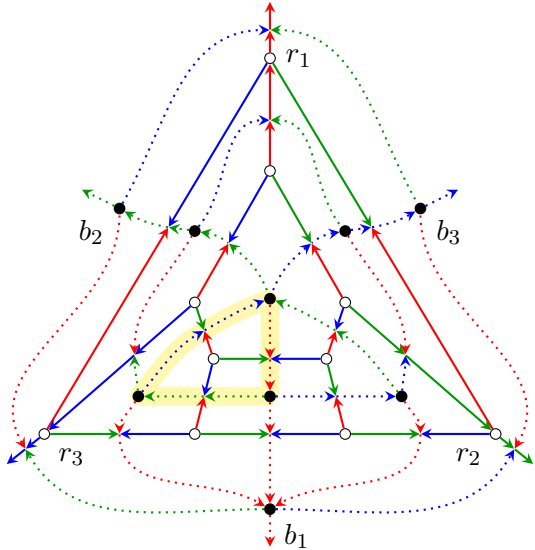
\begin{figure}[!htb]
			\centering
			\begin{tikzpicture}[scale = 0.5]
				\node[rnode, label =right:{$r_1$}] (r1) at (0,10) {};
				\node[rnode, label =-135:{$r_2$}] (r2) at (6,0) {};
				\node[rnode, label =-45:{$r_3$}] (r3) at (-6,0) {};
				\node[rnode] (1) at (-2,0) {};
				\node[rnode] (2) at (2,0) {};
				\node[rnode] (3) at (-1.5,2) {};
				\node[rnode] (4) at (1.5,2) {};
				\node[rnode] (5) at (-2,3.5) {};
				\node[rnode] (7) at (2,3.5) {};
				\node[rnode] (9) at (0,7) {};
				
				\node[snode, label =-45:{$b_1$}] (b1) at(0,-2) {};
				\node[snode, label =190:{$b_2$}] (b2) at(-4,6) {};
				\node[snode, label =-10:{$b_3$}] (b3) at(4,6) {};
				\node[snode] (d1) at(-3.5,1) {};
				\node[snode] (d2) at(0,1) {};
				\node[snode] (d3) at(3.5,1) {};
				\node[snode] (d4) at(-2,5.4) {};
				\node[snode] (d5) at(0,3.6) {};
				\node[snode] (d6) at(2,5.4) {};
				
				\foreach \z/\x/\y in {1/r3/1, 2/1/2, 3/2/r2, 4/1/3, 5/2/4, 8/3/5, 9/4/7, m/3/4}{
					\node[cnode] (c\z) at ($(\x) !0.5! (\y)$) {};
				}
				\foreach \z/\x/\y in {6/r3/5, 7/r2/7}{
					\node[cnode] (c\z) at ($(\x) !0.6! (\y)$) {};
				}
				\foreach \z/\x/\y in {12/r3/r1, 13/r2/r1, 10/9/5, 11/9/7, 14/r1/9}{
					\node[cnode] (c\z) at ($(\x) !0.55! (\y)$) {};
				}

				\foreach \x/\y in {c14/r1, 9/c14, r2/c13, r3/c12, 5/c10, 7/c11, 1/c4, 2/c5, 3/c8, 4/c9}{
					\draw[line, red] (\x) to (\y);
				}
				\foreach \x/\y in {c6/r3, 5/c6, r1/c12, 1/c1, 2/c2, r2/c3, 9/c10, 7/c9, 3/c4, 4/cm}{
					\draw[line, blue] (\x) to (\y);
				}
				\foreach \x/\y in {c7/r2, 7/c7, r1/c13, r3/c1, 1/c2, 2/c3, 9/c11, 5/c8, 3/cm, 4/c5}{
					\draw[line, green] (\x) to (\y);
				}

				\draw [dline, green] (c4) to (d1);
				\draw [dline, green] (d2) to (c4);
				\draw [dline, blue] (d2) to (c5);
				\draw [dline, blue] (c5) to (d3);
				\draw [dline, red] (d1) to [in=90, out=-120] (c1);
				\draw [dline, red] (c1) to [in=150, out=-90] (b1);
				\draw [dline, red] (d2) to (c2);
				\draw [dline, red] (c2) to (b1);
				\draw [dline, red] (d3) to [in=90, out=-60] (c3);
				\draw [dline, red] (c3) to [in=30, out=-90] (b1);
				\draw [dline, green] (d1) to [in=-90, out=95] (c6);
				\draw [dline, red, {stealth}-] (c6) to [in=-120, out=90] (d4);
				\draw [dline, blue] (d3) to [in=-90, out=85] (c7);
				\draw [dline, red, {stealth}-] (c7) to [in=-60, out=90] (d6);
				\draw [dline, blue] (d1) to [in=-150, out=50] (c8);
				\draw [dline, blue] (c8) to [in=-160, out=30] (d5);
				\draw [dline, red, {stealth}-] (d2) to (cm);
				\draw [dline, red, {stealth}-] (cm) to (d5);
				\draw [dline, green] (d3) to [in=-30, out=130] (c9);
				\draw [dline, green] (c9) to [in=-20, out=150] (d5);
				\draw [dline, green, {stealth}-] (d4) to [in=140, out=-10] (c10);
				\draw [dline, green, {stealth}-] (c10) to [in=110, out=-40] (d5);
				\draw [dline, blue, {stealth}-] (d6) to [in=40, out=-170] (c11);
				\draw [dline, blue, {stealth}-] (c11) to [in=70, out=-140] (d5);
				\draw [dline, green] (d4) to [in=-10, out=170] (c12);
				\draw [dline, green] (c12) to [in=-30, out=170] (b2);
				\draw [dline, blue] (d6) to [in=-170, out=10] (c13);
				\draw [dline, blue] (c13) to [in=-150, out=10] (b3);
				\draw [dline, blue] (d4) to [in=180, out=50] (c14);
				\draw [dline, green, {stealth}-] (c14) to [in=130, out=0] (d6);
				
				\node[cnode] (cr3) at ($(r3) !0.5! (-7,-0.8)$) {};
				\node[cnode] (cr1) at ($(r1) !0.5! (0,11.5)$) {};
				\node[cnode] (cr2) at ($(r2) !0.5! (7,-0.8)$) {};
				
				\draw [dline, green] (b1) to [in=-60, out=180] (cr3);
				\draw [dline, red, {stealth}-] (cr3) to [in=-90, out=120] (b2);
				\draw [dline, blue] (b1) to [in=-120, out=0] (cr2);
				\draw [dline, red, {stealth}-] (cr2) to [in=-90, out=60] (b3);		
				\draw [dline, blue] (b2) to [in=-180, out=60] (cr1);
				\draw [dline, green, {stealth}-] (cr1) to [in=120, out=0] (b3);		
				
				\draw[line, red] (r1) to (cr1);
				\draw[line, green] (r2) to (cr2);
				\draw[line, blue] (r3) to (cr3);

				\draw[dline, green] (b2) to (-5,6.5);
				\draw[dline, red] (b1) to  (0,-3);
				
				\draw[line, red] (cr1) to (0,11.5);
				\draw[line, green] (cr2) to  (7,-0.8);
				\draw[line, blue] (cr3) to  (-7,-0.8);
				\draw[dline, blue] (b3) to (5,6.5);

				\begin{pgfonlayer}{background}
					\draw[nline] (d1.center) to [in=-150, out=50] (c8.center) to [in=-160, out=30] (d5.center) to (cm.center) to (d2.center) to (c4.center) -- cycle;
				\end{pgfonlayer}
				
			\end{tikzpicture}
			\caption{The completion of $G$ obtained by superimposing $G^\sigma$ and its suspended dual $G^{\sigma^*}$ (the latter depicted with dotted edges). The primal Schnyder wood is not the minimal element of the lattice of Schnyder woods of $G$, as this completion contains a clockwise directed cycle (marked in yellow).}
			\label{fig:completion}
		\end{figure}
		
		Consider the superposition of $G^\sigma$ and its suspended dual $G^{\sigma^*}$ such that exactly the primal dual pairs of edges cross (here, for every $1 \leq i \leq 3$, the half-edge at $r_i$ crosses the dual edge $b_{i-1}b_{i+1}$).
		
		\begin{definition}\label{def:schnyderdual}
			For any Schnyder wood $S$ of $G^\sigma$, define the orientation and coloring $S^*$ of the suspended dual $G^{\sigma^*}$ as follows (see Figure~\ref{fig:completion}):
			\begin{enumerate}
				\item For every unidirected $(i-1)$-colored edge or half-edge $e$ of $G^\sigma$, color $e^*$ with the two colors $i$ and $i+1$ such that $e$ points to the right of the $i$-colored direction.
				\item Vice versa, for every $i$-$(i+1)$-colored edge $e$ of $G^\sigma$, $(i-1)$-color $e^*$ unidirected such that $e^*$ points to the right of the $i$-colored direction.
				\item For every color $i$, make the half-edge at $b_i$ unidirected, outgoing and $i$-colored.\label{def:schnyderdual3}
			\end{enumerate}
		\end{definition}
		
		The following lemma states that $S^*$ is indeed a Schnyder wood of the suspended dual. By Definition~\ref{def:schnyderdual}\ref{def:schnyderdual3}, the vertices $b_1$, $b_2$ and $b_3$ are the roots of $S^*$.
		
		\begin{lemma}[{\cite{Kant1992}\cite[Prop.~3]{Felsner2004}}]\label{lem:schnyderdual}
			For every Schnyder wood $S$ of $G^\sigma$, $S^*$ is a Schnyder wood of $G^{\sigma^*}$.
		\end{lemma}
		
		Since $S^{*^*} = S$, Lemma~\ref{lem:schnyderdual} gives a bijection between the Schnyder woods of $G^\sigma$ and the ones of $G^{\sigma^*}$.
		Let the \emph{completion} $\widetilde{G}$ of $G$ be the plane graph obtained from the superposition of $G^\sigma$ and $G^{\sigma^*}$ by subdividing each pair of crossing (half-)edges with a new vertex, which we call a \emph{crossing vertex} (see Figure~\ref{fig:completion}). The completion has six half-edges pointing into its outer face.
		
		Any Schyder wood $S$ of $G^\sigma$ implies the following natural orientation and coloring $\widetilde{G}_S$ of its completion $\widetilde{G}$. Let $vw \in E(G^\sigma) \cup E(G^{\sigma^*})$, let $z$ be the crossing vertex of $G^\sigma$ that subdivides $vw$ and consider the coloring of $vw$ in either $S$ or $S^*$. If $vw$ is outgoing of $v$ and $i$-colored, we direct $vz \in E(\widetilde{G})$ toward $z$ and $i$-color it (and do the same for all other vertices than $v$). In the remaining case that $vw$ is unidirected, incoming at $v$ and $i$-colored, we direct $zv \in E(\widetilde{G})$ toward $v$ and $i$-color it. The three half-edges of $G^{\sigma^*}$ inherit the orientation and coloring of $S^*$ for $\widetilde{G}_S$.
		By Definition~\ref{def:schnyderdual}, the construction of $\widetilde{G}_S$ implies immediately the following corollary.
		
		\begin{corollary}\label{cor:crossingvertex}
			Every crossing vertex of $\widetilde{G}_S$ has one outgoing edge and three incoming edges and the latter are colored $\red$, $\green$ and $\blue$ in counterclockwise direction.
		\end{corollary}
		
		Using results on orientations with prescribed outdegrees on the respective completions, Felsner and Mendez~\cite{Mendez1994,Felsner2004a} showed that the set of Schnyder woods of a planar suspension $G^\sigma$ forms a distributive lattice. The order relation of this lattice relates a Schnyder wood of $G^\sigma$ to a second Schnyder wood if the former can be obtained from the latter by reversing the orientation of a directed clockwise cycle in the completion. This gives the following lemma, of which the computational part is due to Fusy~\cite{Fusy2007}.
		
		\begin{lemma}[\cite{Mendez1994,Felsner2004a}\cite{Fusy2007}]\label{lem:latticeminimal}
			For the minimal element $S$ of the lattice of all Schnyder woods of $G^\sigma$, $\widetilde{G}_S$ contains no clockwise directed cycle. Also, $S$ and $\widetilde{G}_S$ can be computed in linear time.
		\end{lemma}
		
		We call the minimal element of the lattice of all Schnyder woods of $G^\sigma$ the \emph{minimal} Schnyder wood of $G^\sigma$. 
		
		\subsection{Ordered path partitions.}
		\begin{definition}\label{def:opp}
			For any $j \in \{1,2,3\}$ and any $\{r_1, r_2, r_3\}$-internally 3-connected plane graph $G$, an \emph{ordered path partition} $\mathcal{P} = (P_0, \ldots, P_s)$ of $G$ \emph{with base-pair} $(r_j,r_{j+1})$ is an ordered partition of $V(G)$ into the vertex sets of induced paths (therefore often referred to as paths) such that the following holds for every $i \in \{0, \ldots, s-1\}$, where $V_i := \bigcup_{q=0}^i V(P_q)$ and the \emph{contour} $C_i$ is the clockwise walk from $r_{j+1}$ to $r_j$ on the outer face of $G[V_i]$. 
			\begin{enumerate}
				\item $P_0$ consists of the vertices of the clockwise path from $r_j$ to $r_{j+1}$ on the outer face boundary, and $P_s = \{r_{j+2}\}$.\label{def:opp_init}
				\item Each vertex in $P_i$ has a neighbor in $V(G) \setminus V_i$.\label{def:opp_at_least_one_neighbor}
				\item $C_i$ is a path.\label{def:opp_2connected}
				\item Each vertex in $C_i$ has at most one neighbor in $P_{i+1}$.\label{def:opp_at_most_one_neighbor}
			\end{enumerate}
		\end{definition}
		
		By Definition~\ref{def:opp}\ref{def:opp_init} and~\ref{def:opp}\ref{def:opp_at_least_one_neighbor}, $G$ contains for every $i$ and every vertex $v \in P_i$ a path from $v$ to $r_{j+2}$ that intersects $V_i$ only in $v$. Since $G$ is plane, we conclude the following.
		
		\begin{lemma}\label{lem:planeconstruction}
			Every path $P_i$ of an ordered path partition is embedded into the outer face of $G[V_{i-1}]$ for every $1 \leq i \leq s$.
		\end{lemma}

		\subsubsection{Compatible Ordered Path Partitions.}
		We describe a connection between Schnyder woods and ordered path partitions that was first given by Badent et al.~\cite[Theorem~5]{Badent2011}. Because a part of its proof was incomplete, the result was then corrected by Alam et al.~\cite[Lemma~1]{Alam2015}, which however outsourced their proof into the extended abstract~\cite[arXiv version, Section~2.2]{Alam2015}.
		
		\begin{definition}
			Let $j \in \{1,2,3\}$ and $S$ be any Schnyder wood of the suspension $G^\sigma$ of $G$. As proven in~\cite[arXiv version, Section~2.2]{Alam2015}, the vertex sets of the inclusion-wise maximal $j$-$(j+1)$-colored paths of $S$ then form an ordered path partition of $G$ with base pair $(r_j,r_{j+1})$ and an order that is a specific linear extension of the partial order given by reachability in the acyclic graph $T^{-1}_j \cup T^{-1}_{j+1} \cup T_{j+2}$; we call this special ordered path partition \emph{compatible} with $S$ and denote it by $\mathcal{P}^{j,j+1}$. 
			\label{def:compatible_opp}
		\end{definition}
		
		For example, for the Schnyder wood given in Figure~\ref{fig:completion}, $\mathcal{P}^{\green,\blue}$ consists of the vertex sets of six maximal \green-\blue-colored paths, of which four are single vertices.
		
		We only need the fact that the order of $\mathcal{P}^{j,j+1}$ is a linear extension of the partial order given by reachability in the acyclic graph $T^{-1}_j \cup T^{-1}_{j+1} \cup T_{j+2}$. We refer the interested reader for further details to \cite{Alam2015,Badent2011}.

		We denote each path $P_i \in \mathcal{P}^{j,j+1}$ by $P_i := \{v^i_1, \dots, v^i_k\}$ such that $v^i_1v^i_2$ is outgoing $j$-colored at $v^i_1$ and, for every $l \in \{1, \ldots, k-1\}$, $v^i_lv^i_{l+1}$ is a $j$-$(j+1)$-colored edge.
		
		Let $C_i$ be as in Definition~\ref{def:opp}. By Definition~\ref{def:opp}\ref{def:opp_2connected} and Lemma~\ref{lem:planeconstruction}, every path $P_i = \{v^i_1, \dots, v^i_k\}$ of an ordered path partition satisfying $i \in \{1, \ldots, s\}$ has a neighbor $v^i_0 \in C_{i-1}$ that is closest to $r_{j+1}$ and a different neighbor $v^i_{k+1} \in C_{i-1}$ that is closest to $r_j$ (see Figure~\ref{fig_cw_cycle_all_edges_right}). We call $v^i_0$ the \emph{left neighbor} of $P_i$, $v^i_{k+1}$ the \emph{right neighbor} of $P_i$ and $P_i^e := \{v^i_0\} \cup P_i \cup \{v^i_{k+1}\}$ the \emph{extension} of $P_i$; we omit superscripts if these are clear from the context. For $0 < i \leq s$, let the path $P_i$ \emph{cover} an edge $e$ or a vertex $x$ if $e$ or $x$ is contained in $C_{i-1}$, but not in $C_i$, respectively.
		
		\begin{lemma}\label{lem:Pi}
			Every path $P_i \neq P_0$ of a compatible ordered path partition $\mathcal{P}^{j,j+1}$ satisfies the following (see Figure~\ref{fig_cw_cycle_all_edges_right}):
			\begin{enumerate}
				\item Every neighbor of $P_i$ that is in $V_{i-1}$ is contained in the path of $C_{i-1}$ between $v^i_0$ and $v^i_{k+1}$.\label{lem:PiA}
				\item $v^i_0v^i_1$ and $v^i_kv^i_{k+1}$ are edges of $G[V_i]$.\label{lem:PiB}
				\item $v^i_0v^i_1$ is $(j+1)$-colored outgoing at $v^i_1$ and $v^i_kv^i_{k+1}$ is $j$-colored outgoing at $v^i_k$.\label{lem:PiC}
				\item Every edge $v^i_lx$ incident to $P_i$ and $V_{i-1}$ except for $v^i_0v^i_1$ and $v^i_kv^i_{k+1}$ is unidirected, directed towards $P_i$ and $(j+2)$-colored and satisfies $x \notin \{v^i_0, v^i_{k+1}\}$.
				\label{lem:PiD}
			\end{enumerate}
		\end{lemma}
		\begin{proof}
			The statement~\ref{lem:PiA} follows directly from Lemma~\ref{lem:planeconstruction} and the definition of left and right neighbor of $P_i$. 
			
			Now, we prove statements~\ref{lem:PiB} and \ref{lem:PiC}. 
			According to Definition~\ref{def:compatible_opp}, the order of $\mathcal{P}^{j,j+1}$ on the vertex sets of paths is a linear extension of the partial order given by reachability in the acyclic graph $T^{-1}_j \cup T^{-1}_{j+1} \cup T_{j+2}$. Hence, we are able to characterize the edges that join $P_i$ with vertices of $V_{i-1}$ and $V-V_i$, respectively. Edges that join $P_i$ with vertices of $V_{i-1}$ are incoming $(j+2)$-colored, unidirected outgoing $j$-colored or unidirected outgoing $(j+1)$-colored at a vertex of $P_i$. Edges that join $P_i$ with vertices of $V - V_i$ are outgoing $(j+2)$-colored, unidirected incoming $j$-colored or unidirected incoming $(j+1)$-colored at a vertex of $P_i$. The remaining edges are the $j$-$(j+1)$-colored edges of $P_i$.  
			
			Let $v^i_ku$ be the outgoing $j$-colored edge at $v^i_k$ and $v^i_1w$ be the outgoing $(j+1)$-colored edge at $v^i_1$. Observe that for $k > 1$, $v^i_1v^i_2$ is outgoing $j$-colored by definition. Thus, as $G[P_i]$ is induced, $w \notin P_i$. If $k=1$, $P_i$ consists of only one vertex and hence $w \notin P_i$. 
			Thus, as $G[P_i]$ is a maximal $j$-$(j+1)$-colored path, $v^i_1w$ is either unidirected $(j+1)$-colored or $(j+1)$-$(j+2)$-colored. As observed above, this implies that $w \in V_{i-1}$. And statement \ref{lem:PiA} yields that $w \in C_{i-1}$. Similarly, we obtain that $u \in C_{i-1}$.
			
			Assume, for the sake of contradiction, that $u$ is closer to $r_{j+1}$ on $C_{i-1}$ than $w$. By definition of $P_i$, for every vertex of $P_i$ the outgoing $j$-colored edge points towards $u$ and the outgoing $(j+1)$-colored edge points towards $w$ on $G[P_i] \cup \{v^i_ku, v^i_1w\}$. By Definition~\ref{def:Schnyderwood}\ref{def:Schnyderwood3}, the outgoing $(j+2)$-colored edge $e$ of a vertex of $P_i$ occurs in the counterclockwise sector from the outgoing $j$-colored to the outgoing $(j+1)$-colored edge excluding both. As $u$ is closer to $r_{j+1}$ on $C_{i-1}$ than $w$, this sector is in the interior of the region bounded by $G[P_i] \cup \{v^i_ku, v^i_1w\}$ and the path from $u$ to $w$ on $C_{i-1}$. Hence, by planarity, $e$ joins $P_i$ with a vertex of $C_{i-1} \subseteq V_{i-1}$, contradicting our above characterization of edges that join $P_i$ with vertices of $V_{i-1}$. Thus, $w$ is closer to $r_{j+1}$ on $C_{i-1}$ than $u$ or $w =u$. If $u = w$, then Lemma~\ref{lem_felsner_no_cycle} is violated by the cycle formed by $P_i \cup u$ in $T_j \cup T_{j+1}^{-1} \cup T_{j+2}^{-1}$, a contradiction. Thus, $w$ is closer to $r_{j+1}$ on $C_{i-1}$ than $u$.
			
			Since $P_i$ is a maximal $j$-$(j+1)$-colored path, for every vertex in $P_i$ the outgoing $j$-colored and the outgoing $(j+1)$-colored edge are either $v^i_ku$, $v^i_1w$ or in $P_i$. Hence, by our above characterization, the edges that join $P_i$ with vertices of $C_{i-1} \subseteq V_{i-1}$ are exactly $v^i_ku$, $v^i_1w$ and the unidirected incoming $(j+2)$-colored edges at vertices of $P_i$. Let $vx$ be such an unidirected incoming $(j+2)$-colored edge with $v \in P_i$. By Definition~\ref{def:Schnyderwood}\ref{def:Schnyderwood3}, $vx$ occurs in the clockwise sector from the outgoing $j$-colored edge to the outgoing $(j+1)$-colored edge around $v$ excluding both. And hence, by planarity and the fact that $w$ is closer to $r_{j+1}$ on $C_{i-1}$ than $u$, $x$ is contained in the path of $C_{i-1}$ from $w$ to $u$. Thus, by the definition of the left and right neighbor $v^i_0$ and $v^i_{k+1}$ of $P_i$ we have $v^i_0 = w$ and $v^i_{k+1} = u$, respectively. The statements~\ref{lem:PiB} and \ref{lem:PiC} follow.
			
			Consider \ref{lem:PiD}. Let $v^i_lx \notin \{v^i_k v^i_{k+1}, v^i_1 v^i_0\}$ be an edge that joins $P_i$ with a vertex $x$ of $V_{i-1}$. By \ref{lem:PiA}, $x \in C_{i-1}$. In the last paragraph, we observed that $v^i_lx$ is incoming $(j+2)$-colored at a vertex of $P_i$. Also, we showed that for every vertex in $P_i$ the outgoing $j$-colored and the outgoing $(j+1)$-edge are either $v^i_k v^i_{k+1}$, $v^i_1 v^i_0$ or in $P_i$. Thus, we obtain that $v^i_lx$ is unidirected incoming $(j+2)$-colored at a vertex of $P_i$. Assume, for the sake of contradiction, that $x = v^i_0$. Then the path from $v^i_l$ to $v^i_1$ on $P_i$, $v^i_0v^i_1$ and $v^i_lv^i_0$ form an oriented cycle in $T_j \cup T^{-1}_{j+1} \cup T^{-1}_{j+2}$, contradicting Lemma~\ref{lem_felsner_no_cycle}. A similar argument shows that $x \neq v^i_{k+1}$.
		\end{proof}
		
		\section{Spanning Trees with Maximum Degree at Most~4}\label{sec:maxdeg4}
		In this section, we prove our main result. The following new lemma on the structure of minimal Schnyder woods and their ordered path partitions is crucial for this proof.
		For $0 < i \leq s$, let the path $P_i$ \emph{cover} an edge $e$ or a vertex $x$ if $e$ or $x$ is contained in $C_{i-1}$, but not in $C_i$, respectively.
		
		\begin{lemma}\label{lem_all_edges_right}
			Let $G$ be a $\sigma$-internally 3-connected plane graph, $S$ be the minimal Schnyder wood of $G^\sigma$ and $\mathcal{P}^{\green,\blue} = (P_0,\dots,P_s)$ be the ordered path partition that is compatible with $S$. Let $P_i := \{v_1,\dots,v_k\} \neq P_0$ be a path of $\mathcal{P}^{\green,\blue}$ and $v_0$ and $v_{k+1}$ be its left and right neighbor. Then every edge $v_lw \notin \{v_0v_1,v_kv_{k+1}\}$ with $v_l \in P_i$ and $w \in V_{i-1}$ is unidirected, \red-colored and incoming at $v_k$ and $w \notin \{v_0, v_{k+1}\}$.
		\end{lemma}
		\begin{proof}
			Consider any edge $v_lw  \notin \{v_0v_1,v_kv_{k+1}\}$ that is incident to $v_l \in P_i$ and $w \in V_{i-1}$ (see Figure~\ref{fig_cw_cycle_all_edges_right}). 
			By Lemma~\ref{lem:Pi}\ref{lem:PiA}, $w$ is either $v_0$, $v_{k+1}$ or a vertex that is covered by $P_i$. As $v_lw \notin \{v_0v_1, v_kv_{k+1}\}$, $v_lw$ must be \red-colored and incoming at $v_l$ and satisfies $w \notin \{v_0,v_{k+1}\}$ by Lemma~\ref{lem:Pi}\ref{lem:PiD}.
			It thus remains to show that $l = k$.
			
			Assume to the contrary that $l \neq k$ and that $v_lw$ is the clockwise first incoming \red-colored edge at $v_l$ (see Figure~\ref{fig_cw_cycle_all_edges_right}). By Corollary~\ref{cor:crossingvertex}, the dual edge of $v_lv_{l+1}$ is unidirected \red-colored in the completion $\widetilde{G}_S$ of $G$; by Corollary~\ref{cor:crossingvertex}, the dual edge of $v_lw$ is \green-\blue-colored. Those dual edges are relate as depicted in Figure~\ref{fig_cw_cycle_all_edges_right}. Hence, $\widetilde{G}_S$ contains the clockwise cycle (see Figure~\ref{fig_cw_cycle_all_edges_right}), which contradicts the assumption that $S$ is the minimal Schnyder wood.
		\end{proof}
		
		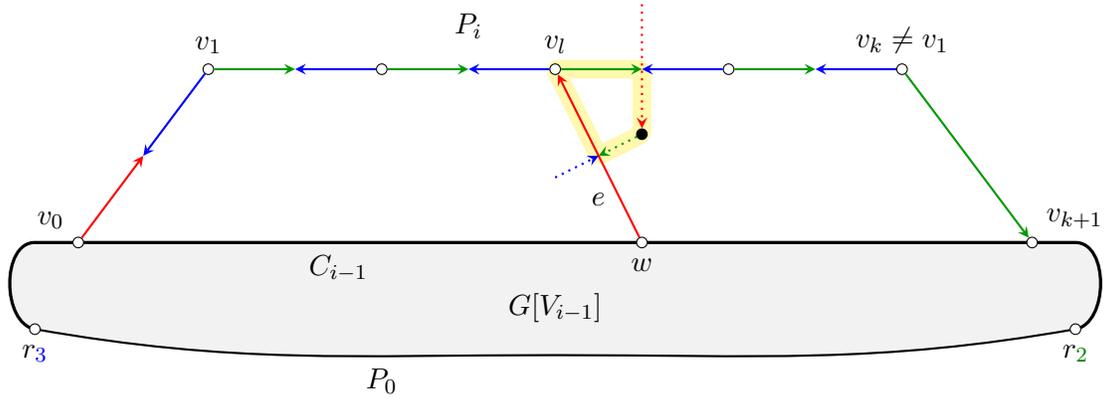
\begin{figure}[!htb]
			\centering
			\begin{tikzpicture}[scale=\textwidth/13cm]
				
				\filldraw[nonarrow, color=black, fill=fillblack] (0,0) to (12,0) to [in=10, out=0] (12,-1) to [in=0, out=-170] (6,-1.3) to [in=-10, out=180] (0,-1) to [in=180, out=170] (0,0);
				\draw[nonarrow, very thick,  color=black] (0,-1) to [in=180, out=170] (0,0) to (12,0) to [in=10, out=0] (12,-1);
				\node[] (labelnode1) at (6,-0.75) {$G[V_{i-1}]$};
				\node[] (labelnode2) at (5,2.5) {$P_i$};
				\node[] (labelnode3) at (6.5,0.5) {$e$};
				\node[] (labelnode3) at (3.5,-0.3) {$C_{i-1}$};
				\node[] (labelnode4) at (4,-1.6) {$P_0$};
				
				\node[rnode, label =below:{$r_\blue$}] (r3) at (0,-1) {};
				\node[rnode, label =below:{$r_\green$}] (r2) at (12,-1) {};
				
				\node[rnode, label =135:{$v_0$}] (v0) at (0.5,0) {};
				\node[rnode, label =above:{$v_1$}] (v1) at (2,2) {};
				\node[rnode] (v2) at (4,2) {};
				\node[rnode, label=above:{$v_l$}] (v3) at (6,2) {};
				\node[rnode] (v4) at (8,2) {};
				\node[rnode, label=above:{$v_k \neq v_1$}] (v5) at (10,2) {};
				\node[rnode, label =45:{$v_{k+1}$}] (v6) at (11.5,0) {};
				
				\foreach \x/\y in {v1/v2, v2/v3, v3/v4, v4/v5}{
					\draw[line, green] (\x) to ($(\x) !0.5! (\y)$);
					\draw[line, blue] (\y) to ($(\x) !0.5! (\y)$);
				}
				
				\draw[line, blue] (v1) to ($(v1) !0.5! (v0)$);
				\draw[line, red] (v0) to ($(v1) !0.5! (v0)$);
				
				\draw[line, green] (v5) to (v6);
				
				\node[rnode, label=below:{$w$}] (x) at (7,0) {};
				\draw[line, red] (x) to (v3);
				
				\node[snode] (y) at (7, 1.25) {};
				\draw[dline, green] (y) to ($(x) !0.5! (v3)$);
				\draw[dline, blue] (6, 0.75) to ($(x) !0.5! (v3)$);
				\draw[dline, red] (7,2.75) to (y);

				\begin{pgfonlayer}{background}
					\draw[nline] (7,2) to (y.center) to ($(x) !0.5! (v3)$) to (v3.center) -- cycle;
				\end{pgfonlayer}
				
			\end{tikzpicture}
			\caption{The clockwise cycle of $\widetilde{G}_S$ of the proof of Lemma~\ref{lem_all_edges_right}, depicted in yellow. }
			\label{fig_cw_cycle_all_edges_right}
		\end{figure}
		
		For a spanning subgraph $T$ of a plane graph $G$, let the \emph{co-graph} $\neg T^*$ be the spanning subgraph $(V^*,(E(G)-E(T))^*)$ of $G^*$. As stated in the introduction, $\neg T^*$ is a spanning tree if $T$ is one and in that case called a co-tree.
		
		\begin{theorem}
			Every $\{r_1, r_2, r_3\}$-internally 3-connected plane graph $G$ contains a 4-tree $T$ whose co-tree $\neg T^*$ is a 4-tree.
			\label{thm_max_deg_4}
		\end{theorem}
		\begin{proof}
			We sketch the general idea of the proof: First, we identify a spanning candidate graph $H \subseteq G$ such that $\neg H^*$ is a subgraph of $G^*$ that has the same structural properties as $H$. We then define a subset $D$ of the edges of $H$ such that $H-D$ is acyclic and $\neg H^* + D^*$ has maximum degree 4. We use the same arguments to define a similar subset $D'$ for $\neg H^*$. In the end, we need to show that $D'^*$ and $D^*$ do not create new cycles in $\neg H^*$ and $H$, respectively. That way we obtain that the co-graph of $H-D+D'^*$ is $\neg H^* -D' + D^*$, and both graphs are acyclic and of maximum degree~4. Since a spanning subgraph $G'$ of $G$ is connected if and only if $G-E(G')$ does not contain any edge cut of $G$, the cut-cycle duality~\cite[Prop. 4.6.1]{Diestel2012} proves that those two graphs are both connected, which gives the claim.
			
			Let $S$ be the minimal Schnyder wood of $G^{\sigma}$. By Lemma~\ref{lem:latticeminimal}, the completion $\widetilde{G}_S$ of $G$ contains no clockwise directed cycle. Since $\widetilde{G}_S$ contains the completion of the suspended dual $G^{\sigma^*}$ except for its three outer vertices (which do not affect clockwise cycles), $S^*$ is a minimal Schnyder wood of $G^{\sigma^*}$.
			
			Let $H$ be the spanning subgraph of $G$ whose edge set consists of the bidirected edges of $S$. Recall that an edge $e \in E(G)$ is not in $H$ if and only if $e^*$ is in $\neg H^*$. By Definition~\ref{def:schnyderdual}, $\neg H^*$ contains therefore exactly the bidirected edges of $S^*$, except for the three bidirected edges on the outer face boundary of $G^{\sigma^*}$, as these are not dual edges of $G$ (in fact, these three edges appear only in the suspended dual $G^{\sigma^*}$ and were necessary to define dual Schnyder woods).
			
			Since every vertex is incident to at most three bidirected edges by Definition~\ref{def:Schnyderwood}\ref{def:Schnyderwood3} for $S$ and as well for $S^*$, both $H$ and $\neg H^*$ have maximum degree at most three. However, $H$ and $\neg H^*$ may neither be connected nor acyclic. In fact, $H$ contains always the outer face boundary of $G$ as a cycle, as all edges are bidirected by the definition of the first paths of the compatible ordered path partitions $\mathcal{P}^{\red,\green}$, $\mathcal{P}^{\green,\blue}$ and $\mathcal{P}^{\blue,\red}$.
			
			We will therefore iteratively identify edges of cycles of $H$ such that $\neg H^*$ still has maximum degree at most four when those cycles are deleted in $H$.
			In order to do this, we iteratively define edges $D$ and $D'$ that are deleted from $H$ and $\neg H^*$, starting with $D := D' := \emptyset$.

			Let $C$ be a cycle of $H$ and let $(P_0, \dots, P_s)$ be the paths of the compatible ordered path partition $\mathcal{P}^{\green,\blue}$ of $S$. Let $P$ be the path of maximal length in $C$ such that $P \subseteq P_M$ with $M := \max \{i \mid P_i \cap V(C) \neq \emptyset \}$; we call $P$ the \emph{index maximal subpath} of $C$, as it is the fraction of $C$ highest up in the order of $\mathcal{P}^{\green,\blue}$. Since $C$ has only bidirected edges, the statement of Lemma~\ref{lem_all_edges_right} about $e$ being unidirected implies that $P = P_M$ and that $C$ contains the extension of $P$; in particular, $P \in \mathcal{P}^{\green,\blue}$.
			
			Denote by $\mathcal{P}_{max}$ the set of index maximal subpaths of all cycles of $H$.
			For a path $P \in \mathcal{P}_{max} \setminus \{P_s\}$, let $P_L$ with $L := \min\{ i \mid P_i \text{ covers an edge of the extension of $P$} \}$ be the \emph{minimal-covering path} of $P$ (recall that this extension is part of the cycle and the minimal-covering path exists, as $P_s$ is excluded). Denote by $\mathcal{P}_{cover}$ the set of the minimal-covering paths of all index maximal subpaths in $\mathcal{P}_{max} \setminus \{P_s\}$. In particular, $P_s = r_\red$ is the index maximal subpath of the outer face boundary of $G$, which is a bidirected cycle, as shown before. Since no edge of the extension of $P_s$ is covered by another path of $\mathcal{P}^{\green,\blue}$, we add the outgoing \green-colored edge of $r_\red$ to $D$ in order to destroy the outer face cycle.
			
			Next, we process the paths of $\mathcal{P}_{cover}$ in reverse order of $\mathcal{P}^{\green,\blue}$, i.e. from highest to lowest index. Let $P_c = \{v_1, \ldots, v_k\} \in \mathcal{P}_{cover}$ for some $c \in \{1,\ldots, s-1\}$ be the path under consideration. Let $P'_1, \ldots, P'_l$ be the index maximal paths for which $P_c$ is the minimal-covering path, ordered clockwise around the outer face of $G[V_{c-1}]$ (see Figure~\ref{fig_illustration_for_definition}). Let $f_1, \ldots, f_a$ be the faces incident to $v_k$ in counterclockwise order from the outgoing \blue-colored edge to the outgoing \green-colored edge; we say that $f_1, \ldots, f_a$ are \emph{below} $P_c$. For every path of $\{P'_1, \ldots, P'_l\}$, we will add an edge to $D$ that is on the extension of that path. Thus, after having processed every path in $\mathcal{P}_{cover}$ in this way, a cycle in $H$ does not exist in $H-D$ anymore.
			
			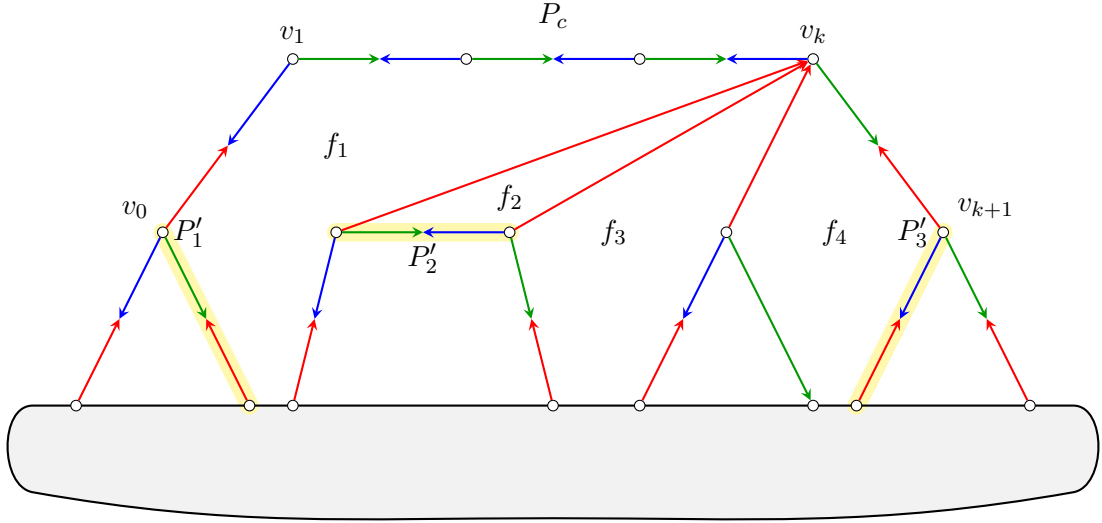
\begin{figure}[!htb]
				\centering
				\begin{tikzpicture}[scale=\textwidth/13cm]
					
					\node[] (labelnode2) at (6,4.5) {$P_c$};
					
					\node[] (labelnode2) at (1.8,2) {$P_1'$};
					\node[] (labelnode2) at (4.5,1.7) {$P_2'$};
					\node[] (labelnode2) at (10.15,2) {$P_3'$};
					
					\node[] (labelnode2) at (3.5,3) {$f_1$};
					\node[] (labelnode2) at (5.5,2.4) {$f_2$};
					\node[] (labelnode2) at (6.7,2) {$f_3$};
					\node[] (labelnode2) at (9.25,2) {$f_4$};

					\node[rnode, label =135:{$v_0$}] (v0) at (1.5,2) {};
					\node[rnode, label =above:{$v_1$}] (v1) at (3,4) {};
					\node[rnode] (v2) at (5,4) {};
					\node[rnode] (v3) at (7,4) {};
					\node[rnode, label=above:{$v_k$}] (v4) at (9,4) {};
					\node[rnode, label =45:{$v_{k+1}$}] (v5) at (10.5,2) {};
					
					\node[rnode] (w0) at (0.5,0) {};
					\node[rnode] (w1) at (2.5,0) {};
					\node[rnode] (w2) at (3,0) {};
					\node[rnode] (w3) at (6,0) {};
					\node[rnode] (w4) at (7,0) {};
					\node[rnode] (w5) at (9,0) {};
					\node[rnode] (w6) at (9.5,0) {};
					\node[rnode] (w7) at (11.5,0) {};
					
					\node[rnode] (i0) at (3.5,2) {};
					\node[rnode] (i1) at (5.5,2) {};
					\node[rnode] (i2) at (8,2) {};
					
					\foreach \x/\y in {v1/v2, v2/v3, v3/v4, i0/i1}{
						\draw[line, green] (\x) to ($(\x) !0.5! (\y)$);
						\draw[line, blue] (\y) to ($(\x) !0.5! (\y)$);
					}
					\foreach \x in {i0, i1, i2}{
						\draw[line, red] (\x) to (v4);
					}
					\foreach \x/\y in {w0/v0, w2/i0, w4/i2, w6/v5}{
						\draw[line, red] (\x) to ($(\x) !0.5! (\y)$);
						\draw[line, blue] (\y) to ($(\x) !0.5! (\y)$);
					}
					\foreach \x/\y in {w1/v0, w3/i1, w7/v5}{
						\draw[line, red] (\x) to ($(\x) !0.5! (\y)$);
						\draw[line, green] (\y) to ($(\x) !0.5! (\y)$);
					}
					
					\draw[line, green] (i2) to (w5);
					
					\draw[line, blue] (v1) to ($(v1) !0.5! (v0)$);
					\draw[line, red] (v0) to ($(v1) !0.5! (v0)$);
					
					\draw[line, green] (v4) to ($(v4) !0.5! (v5)$);
					\draw[line, red] (v5) to ($(v4) !0.5! (v5)$);
					
					\begin{pgfonlayer}{background}
						\filldraw[nonarrow, color=white, fill=fillblack] (0,0) to (12,0) to [in=10, out=0] (12,-1) to [in=0, out=-170] (6,-1.3) to [in=-10, out=180] (0,-1) to [in=180, out=170] (0,0);
						\foreach \x/\y in {v0/w1, i0/i1, v5/w6}{
							\draw[nline, line cap=round] (\x.center) to (\y.center);
						}
						\draw[nonarrow, color=black] (0,0) to (12,0) to [in=10, out=0] (12,-1) to [in=0, out=-170] (6,-1.3) to [in=-10, out=180] (0,-1) to [in=180, out=170] (0,0);
					\end{pgfonlayer}
					
				\end{tikzpicture}
				\caption{Illustration for some of the definitions used in Theorem~\ref{thm_max_deg_4}. If Case~1 applies to $P_c$, we add the edges marked in yellow to $D$.}
				\label{fig_illustration_for_definition}
			\end{figure}
			
			\begin{figure}[!htb]
				\centering
				\begin{tikzpicture}[scale=\textwidth/13cm]
					
					\filldraw[nonarrow, color=black, fill=fillblack] (0,0) to (12,0) to [in=10, out=0] (12,-1) to [in=0, out=-170] (6,-1.3) to [in=-10, out=180] (0,-1) to [in=180, out=170] (0,0);
					\node[] (labelnode2) at (6,4.5) {$P_c$};
					
					\node[] (labelnode2) at (10.15,2) {$P_l'$};
					
					\node[rnode, label =135:{$v_0$}] (v0) at (1.5,2) {};
					\node[rnode, label =above:{$v_1$}] (v1) at (3,4) {};
					\node[rnode] (v2) at (5,4) {};
					\node[rnode] (v3) at (7,4) {};
					\node[rnode, label=above:{$v_k$}] (v4) at (9,4) {};
					\node[rnode, label =45:{$v_{k+1}=w_1$}] (v5) at (10.5,2) {};
					
					\node[rnode] (w0) at (0.5,0) {};
					\node[rnode] (w1) at (2.5,0) {};
					\node[rnode] (w2) at (3,0) {};
					\node[rnode] (w3) at (6,0) {};
					\node[rnode, label =below:{$w_0$}] (w6) at (9.5,0) {};
					\node[rnode] (w7) at (11.5,0) {};
					
					\node[rnode] (i0) at (3.5,2) {};
					\node[rnode] (i1) at (5.5,2) {};
					
					\node[cnode] (v4v5) at ($(v4) !0.5! (v5)$) {};
					\node[cnode] (w6v5) at ($(w6) !0.5! (v5)$) {};
					
					\node[snode, label =180:{$f_a^*$}] (fa) at (9,2) {};
					
					\node[cnode] (ro) at ($(fa) !1.5! (v4v5)$) {};
					\node[cnode] (go) at ($(fa) !1.5! (w6v5)$) {};

					\foreach \x/\y in {v1/v2, v2/v3, v3/v4, i0/i1}{
						\draw[line, green] (\x) to ($(\x) !0.5! (\y)$);
						\draw[line, blue] (\y) to ($(\x) !0.5! (\y)$);
					}
					\foreach \x in {i0, i1}{
						\draw[line, red] (\x) to (v4);
					}
					\foreach \x/\y in {w0/v0, w2/i0, w6/v5}{
						\draw[line, red] (\x) to ($(\x) !0.5! (\y)$);
						\draw[line, blue] (\y) to ($(\x) !0.5! (\y)$);
					}
					\foreach \x/\y in {w1/v0, w3/i1, w7/v5}{
						\draw[line, red] (\x) to ($(\x) !0.5! (\y)$);
						\draw[line, green] (\y) to ($(\x) !0.5! (\y)$);
					}
					
					\draw[line, blue] (v1) to ($(v1) !0.5! (v0)$);
					\draw[line, red] (v0) to ($(v1) !0.5! (v0)$);
					
					\draw[line, green] (v4) to ($(v4) !0.5! (v5)$);
					\draw[line, green] ($(v4) !0.5! (v5)$) to (v5);
					
					\draw[dline, green] (w6v5) to (fa);
					\draw[dline, blue] (fa) to (v4v5);
					\draw[dline, red] (ro) to (v4v5);
					\draw[dline, green] (go) to (w6v5);

					\begin{pgfonlayer}{background}
						\draw[nline] (fa.center) to (v4v5.center) to (v5.center) to (w6v5.center) -- cycle;
					\end{pgfonlayer}
					
				\end{tikzpicture}
				\caption{If $v_kv_{k+1}$ is \green-colored, then $\widetilde{G}_S$ contains a clockwise cycle (depicted in yellow).}
				\label{fig_cw_cycle_if_1col_23_col}
			\end{figure}
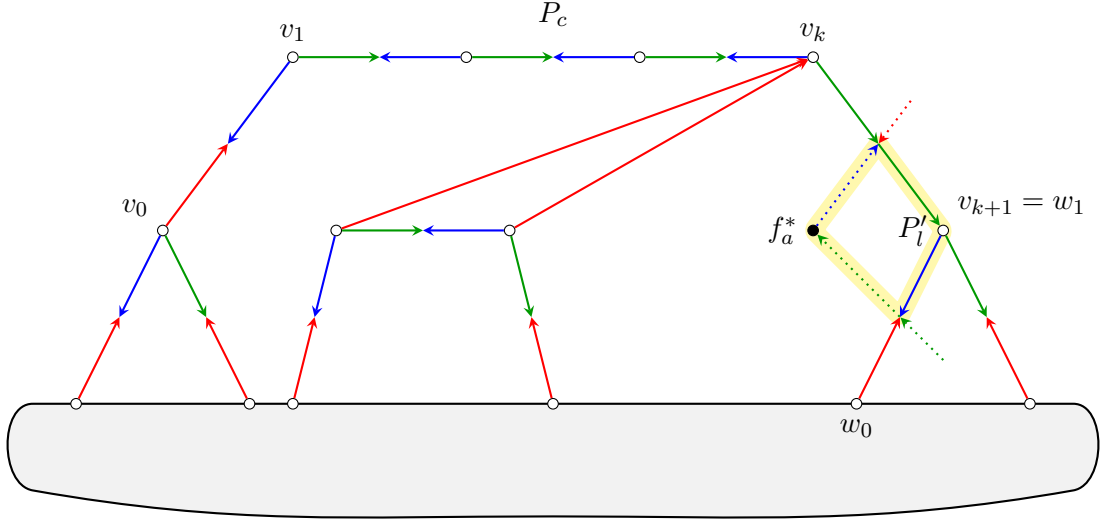
			
			Consider the case that $v_{k+1} = w_1$ for a path $P'_l = \{w_1, \ldots, w_t\}$. Assume to the contrary that then $v_kv_{k+1}$ is not \red-\green-colored. Since $P'_l$ is an index maximal subpath, $w_0w_1$ is \red-\blue-colored. By Lemma~\ref{lem:Pi}\ref{lem:PiC}, then $v_kv_{k+1}$ is unidirected \green-colored. By Corollary~\ref{cor:crossingvertex}, this implies that $(v_kv_{k+1})^*$ is \red-\blue-colored. Hence, $\widetilde{G}_S$ contains the clockwise cycle in Figure~\ref{fig_cw_cycle_if_1col_23_col}, which contradicts the assumption that $S$ is the minimal Schnyder wood. We conclude that $v_kv_{k+1}$ is \red-\green-colored in that case.
			
			Now, we select one edge from each of the extensions of the paths $P_1', \ldots, P_l'$ and add it to $D$.  We select those edges that have smallest possible impact on the maximum degree of the dual graph. Thus edges of the $P_1', \ldots, P_l'$ themselves that are covered by $P_c$ are always preferable (see Figure~\ref{fig_illustration_for_definition}). In Figure~\ref{fig_illustration_for_definition}, the edge of $P'_2$ causes a higher degree at a dual vertex below $P'_2$ and at $f^*_2$, but $f_2$ is a triangle and thus the degree of $f^*_2$ never exceeds 3. If for example $v_0v_1 \in D$, this raises the degree of $f^*_1$. Thus, we try to pick edges that are not incident to $f^*_1$, i.e. if we cannot choose an edge of a path itself, we choose the edge to its right neighbor. This motivation results in the following procedure. We distinguish two cases.
			\begin{description}
				\item[Case 1:] $P_c$ is not an index maximal subpath (see Figure~\ref{fig_illustration_for_definition}).\\
				For every $i \in \{1, \ldots, l\}$, if $P_c$ covers an edge of $G[P_i']$, then we add one such edge to $D$. If for $P_l' = \{w_1, \ldots, w_t\}$ we have $w_1 = v_{k+1}$ (note that this excludes the previous condition), then we add $w_0w_1$ to $D$. For all remaining $i \in \{1, \ldots, l\}$ for which none of the above conditions apply, we set $P_i' = \{u_1, \ldots, u_t\}$ and add the edge $u_tu_{t+1}$ to $D$.
				\item[Case 2:] $P_c$ is an index maximal subpath.\\
				Since the minimal-covering path of $P_c$ has higher index than $P_c$ itself, there already is either an edge of $G[P_c]$, $v_0v_1$ or $v_kv_{k+1}$ in $D$.
				\begin{description}
					\item[Case 2.1:] An edge of $G[P_c]$ or $v_0v_1$ is in $D$ (see Figure~\ref{fig_Pc_path}).\\
					We proceed as in Case~1.
					\item[Case 2.2:] $v_kv_{k+1} \in D$ (see Figure~\ref{fig_Pc_singleton_and_ind_max})\\
					For every $i \in \{1, \ldots, l\}$, if $P_c$ covers an edge of $G[P_i']$, then we add one such edge to $D$. If for $P_1' = \{p_1, \ldots, p_b\}$ we have $p_b = v_0$ (note that this excludes the previous condition), then we add $p_bp_{b+1}$ to $D$. For all remaining $i \in \{1, \ldots, l\}$ for which none of the above conditions apply, we set $P_i' = \{u_1, \ldots, u_t\}$ and add the edge $u_0u_1$ to $D$.
				\end{description}
			\end{description}
			
			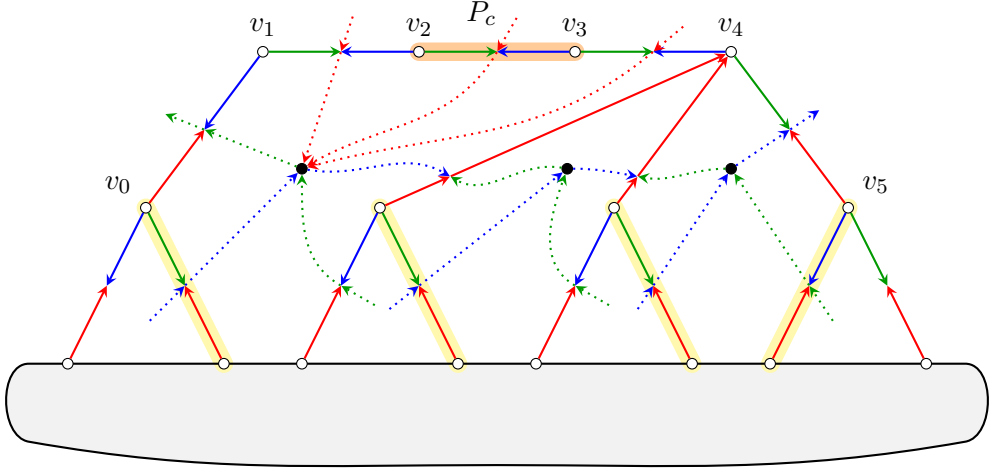
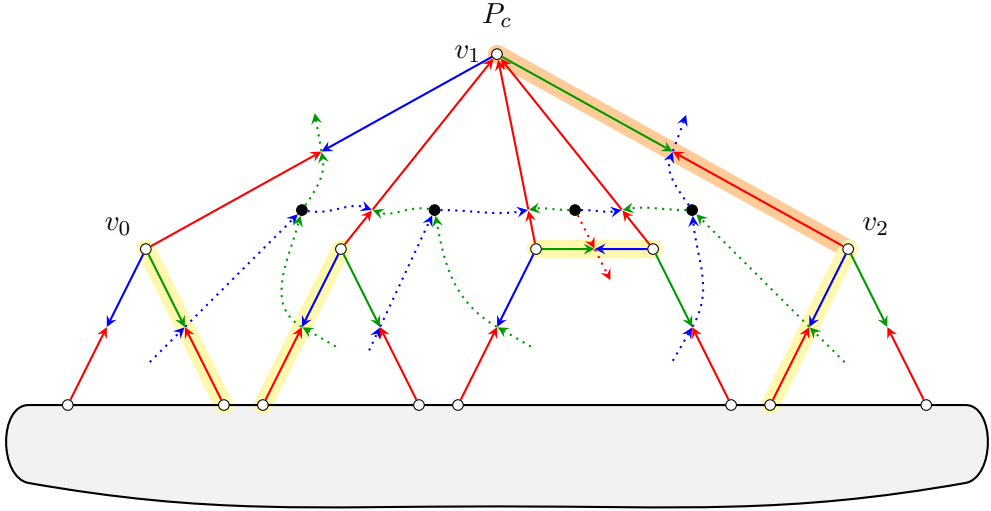
\begin{figure}[!htb]
				\centering
				\begin{subfigure}{.9\textwidth}
					\centering
					\begin{tikzpicture}[scale=\textwidth/13cm]

						\node[] (labelnode2) at (5.8,4.5) {$P_c$};

						\node[rnode, label =135:{$v_0$}] (v0) at (1.5,2) {};
						\node[rnode, label =above:{$v_1$}] (v1) at (3,4) {};
						\node[rnode, label =above:{$v_2$}] (v2) at (5,4) {};
						\node[rnode, label =above:{$v_3$}] (v3) at (7,4) {};
						\node[rnode, label=above:{$v_4$}] (v4) at (9,4) {};
						\node[rnode, label =45:{$v_{5}$}] (v5) at (10.5,2) {};
						
						\node[rnode] (w0) at (0.5,0) {};
						\node[rnode] (w1) at (2.5,0) {};
						\node[rnode] (w2) at (3.5,0) {};
						\node[rnode] (w3) at (5.5,0) {};
						\node[rnode] (w4) at (6.5,0) {};
						\node[rnode] (w5) at (8.5,0) {};
						\node[rnode] (w6) at (9.5,0) {};
						\node[rnode] (w7) at (11.5,0) {};
						
						\node[rnode] (i0) at (4.5,2) {};
						\node[rnode] (i1) at (7.5,2) {};
						
						\node[snode] (f1) at (3.5,2.5) {};
						\node[snode] (f2) at (6.9,2.5) {};
						\node[snode] (f3) at (9,2.5) {};
						
						\node[cnode] (i0v4) at ($(i0) !0.2! (v4)$) {};
						\node[cnode] (i1v4) at ($(i1) !0.2! (v4)$) {};
						
						\foreach \x/\y in {v0/v1, v1/v2, v2/v3, v3/v4, v4/v5, w1/v0, w2/i0, w3/i0, w4/i1, w5/i1, w6/v5}{
							\node[cnode] (\x\y) at ($(\x) !0.5! (\y)$) {};
						}
						
						\foreach \x/\y in {v1/v2, v2/v3, v3/v4}{
							\draw[line, green] (\x) to ($(\x) !0.5! (\y)$);
							\draw[line, blue] (\y) to ($(\x) !0.5! (\y)$);
						}
						\foreach \x in {i0, i1}{
							\draw[line, red] (\x) to (\x v4);
							\draw[line, red] (\x v4) to (v4);
						}
						\foreach \x/\y in {w0/v0, w2/i0, w4/i1, w6/v5}{
							\draw[line, red] (\x) to ($(\x) !0.5! (\y)$);
							\draw[line, blue] (\y) to ($(\x) !0.5! (\y)$);
						}
						\foreach \x/\y in {w1/v0, w3/i0, w5/i1, w7/v5, v5/v4}{
							\draw[line, red] (\x) to ($(\x) !0.5! (\y)$);
							\draw[line, green] (\y) to ($(\x) !0.5! (\y)$);
						}
						
						\draw[line, blue] (v1) to ($(v1) !0.5! (v0)$);
						\draw[line, red] (v0) to ($(v1) !0.5! (v0)$);
						
						\draw[dline, blue] (w1v0) to (f1);
						\draw[dline, green] (w2i0) to [out=150, in=-90] (f1);
						\draw[dline, green] (w2i0)+(-30:0.5) to (w2i0);
						\draw[dline, blue] (w3i0) to (f2);
						\draw[dline, green] (w4i1) to [out=150, in=-110] (f2);
						\draw[dline, green] (w4i1)+(-30:0.5) to (w4i1);
						\draw[dline, blue] (w5i1) to (f3);
						\draw[dline, green] (w6v5) to (f3);
						
						\draw[dline, red] (v1v2) to (f1);
						\draw[dline, red] (v2v3) to [out=-120, in=30] (f1);
						\draw[dline, red] (v2v3)+(60:0.5) to (v2v3);
						\draw[dline, red] (v3v4) to [out=-140, in=15] (f1);
						\draw[dline, red] (v3v4)+(40:0.5) to (v3v4);
						
						\draw[dline, green] (f1) to (v0v1);
						\draw[dline, blue] (f3) to (v4v5);
						
						\draw[dline, blue] ($(f1) !1.3! (w1v0)$) to (w1v0);
						\draw[dline, blue] ($(f2) !1.2! (w3i0)$) to (w3i0);
						\draw[dline, blue] ($(f3) !1.2! (w5i1)$) to (w5i1);
						\draw[dline, green] ($(f3) !1.3! (w6v5)$) to (w6v5);
						
						\draw[dline, green] (v0v1) to ($(f1) !1.4! (v0v1)$);
						\draw[dline, blue] (v4v5) to ($(f3) !1.5! (v4v5)$);
						
						\draw[dline, red] ($(f1) !1.3! (v1v2)$) to (v1v2);
						
						\draw[dline, blue] (f1) to [out=-10, in=150] (i0v4);
						\draw[dline, green] (f2) to [out=170, in=-30] (i0v4);
						\draw[dline, blue] (f2) to [out=0, in=170] (i1v4);
						\draw[dline, green] (f3) to [out=180, in=-10] (i1v4);
						
						\begin{pgfonlayer}{background}
							\filldraw[nonarrow, color=white, fill=fillblack] (0,0) to (12,0) to [in=10, out=0] (12,-1) to [in=0, out=-170] (6,-1.3) to [in=-10, out=180] (0,-1) to [in=180, out=170] (0,0);
							\foreach \x/\y in {w1/v0, w3/i0, w5/i1, w6/v5}{
								\draw[nline, line cap=round] (\x.center) to (\y.center);
							}
							\draw[n2line, line cap=round] (v2.center) to (v3.center);
							\draw[nonarrow, color=black] (0,0) to (12,0) to [in=10, out=0] (12,-1) to [in=0, out=-170] (6,-1.3) to [in=-10, out=180] (0,-1) to [in=180, out=170] (0,0);
						\end{pgfonlayer}
						
					\end{tikzpicture}
					\caption{The situation in Case~2.1. Here the edge $v_2v_3$ is marked in orange and in $D$ before we consider $P_c$. The edges that we add to $D$ are marked in yellow.}
					\label{fig_Pc_path}
				\end{subfigure}
				\begin{subfigure}{.9\textwidth}
					\centering
					\begin{tikzpicture}[scale=\textwidth/13cm]
						\node[] (labelnode2) at (6,5) {$P_c$};
						
						\node[rnode, label =135:{$v_0$}] (v0) at (1.5,2) {};
						\node[rnode, label =left:{$v_1$}] (v1) at (6,4.5) {};
						\node[rnode, label =45:{$v_{2}$}] (v5) at (10.5,2) {};
						
						\node[rnode] (w0) at (0.5,0) {};
						\node[rnode] (w1) at (2.5,0) {};
						\node[rnode] (w2) at (3,0) {};
						\node[rnode] (w3) at (5,0) {};
						\node[rnode] (w4) at (5.5,0) {};
						\node[rnode] (w5) at (9,0) {};
						\node[rnode] (w6) at (9.5,0) {};
						\node[rnode] (w7) at (11.5,0) {};
						
						\node[rnode] (i0) at (4,2) {};
						\node[rnode] (i1) at (6.5,2) {};
						\node[rnode] (i2) at (8,2) {};
						
						\node[snode] (f1) at (3.5,2.5) {};
						\node[snode] (f2) at (5.2,2.5) {};
						\node[snode] (f3) at (7,2.5) {};
						\node[snode] (f4) at (8.5,2.5) {};
						
						\node[cnode] (i0v1) at ($(i0) !0.2! (v1)$) {};
						\node[cnode] (i1v1) at ($(i1) !0.2! (v1)$) {};
						\node[cnode] (i2v1) at ($(i2) !0.2! (v1)$) {};
						
						\foreach \x/\y in {v0/v1, v1/v5, w1/v0, w2/i0, w3/i0, w4/i1, w5/i2, i1/i2, w6/v5}{
							\node[cnode] (\x\y) at ($(\x) !0.5! (\y)$) {};
						}
						
						\foreach \x/\y in {i1/i2}{
							\draw[line, green] (\x) to ($(\x) !0.5! (\y)$);
							\draw[line, blue] (\y) to ($(\x) !0.5! (\y)$);
						}
						\foreach \x in {i0, i1, i2}{
							\draw[line, red] (\x) to (\x v1);
							\draw[line, red] (\x v1) to (v1);
						}
						\foreach \x/\y in {w0/v0, w2/i0, w4/i1, w6/v5, v0/v1}{
							\draw[line, red] (\x) to ($(\x) !0.5! (\y)$);
							\draw[line, blue] (\y) to ($(\x) !0.5! (\y)$);
						}
						\foreach \x/\y in {w1/v0, w3/i0, w5/i2, w7/v5, v5/v1}{
							\draw[line, red] (\x) to ($(\x) !0.5! (\y)$);
							\draw[line, green] (\y) to ($(\x) !0.5! (\y)$);
						}

						\draw[dline, blue] (w1v0) to (f1);
						\draw[dline, blue] ($(f1) !1.3! (w1v0)$) to (w1v0);
						\draw[dline, green] (w2i0) to [out=150, in=-105] (f1);
						\draw[dline, green] (w2i0)+(-30:0.5) to (w2i0);
						\draw[dline, blue] (w3i0) to (f2);
						\draw[dline, blue] ($(f2) !1.2! (w3i0)$) to (w3i0);
						\draw[dline, green] (w4i1) to [out=150, in=-80] (f2);
						\draw[dline, green] (w4i1)+(-30:0.5) to (w4i1);
						\draw[dline, blue] (w5i2) to [out=60, in=-85] (f4);
						\draw[dline, blue] (w5i2)+(-120:0.5) to (w5i2);
						\draw[dline, green] (w6v5) to (f4);
						\draw[dline, green] ($(f4) !1.3! (w6v5)$) to (w6v5);
						
						\draw[dline, green] (f1) to [out=55, in=-80] (v0v1);
						\draw[dline, green] (v0v1) to +(100:0.5);
						\draw[dline, blue] (f4) to [out=125, in=-110] (v1v5);
						\draw[dline, blue] (v1v5) to +(70:0.5);
						
						\draw[dline, blue] (f1) to [out=-10, in=160] (i0v1);
						\draw[dline, green] (f2) to [out=170, in=-20] (i0v1);
						\draw[dline, blue] (f2) to [out=0, in=-170] (i1v1);
						\draw[dline, green] (f3) to [out=180, in=10] (i1v1);
						\draw[dline, blue] (f3) to [out=0, in=-170] (i2v1);
						\draw[dline, green] (f4) to [out=180, in=10] (i2v1);
						
						\draw[dline, red] (f3) to (i1i2);
						\draw[dline, red] (i1i2) to ($(f3) !1.8! (i1i2)$);
						
						\begin{pgfonlayer}{background}
							\filldraw[nonarrow, color=white, fill=fillblack] (0,0) to (12,0) to [in=10, out=0] (12,-1) to [in=0, out=-170] (6,-1.3) to [in=-10, out=180] (0,-1) to [in=180, out=170] (0,0);
							
							\draw[n2line, line cap=round] (v1.center) to (v5.center);
							
							\foreach \x/\y in {w1/v0, w2/i0, i1/i2, w6/v5}{
								\draw[nline, line cap=round] (\x.center) to (\y.center);
							}
							
							\draw[nonarrow, color=black] (0,0) to (12,0) to [in=10, out=0] (12,-1) to [in=0, out=-170] (6,-1.3) to [in=-10, out=180] (0,-1) to [in=180, out=170] (0,0);
						\end{pgfonlayer}
						
					\end{tikzpicture}
					\caption{The situation in Case~2.2. The edge $v_1v_2$ is marked in orange and in $D$ before we consider $P_c$. The edges that we then add to $D$ are marked in yellow.}
					\label{fig_Pc_singleton_and_ind_max}
				\end{subfigure}
				\caption{Subcases for which $P_c$ is an index maximal subpath in Theorem~\ref{thm_max_deg_4}.}
				\label{fig_illustration_for_proof}
			\end{figure}
			
			We now need to show that the maximum degree of $\neg H^*+D^*$ does not exceed 4. We now prove that, after having processed $P_c$, no more boundary edges of any $f \in \{f_1, \ldots , f_a\}$ are added to $D$: Assume to the contrary that there is a face $f \in \{f_1, \ldots , f_a\}$ and an edge $e$ on the boundary of $f$ such that $e$ is not in $D$ after having processed $P_c$ but will be added later. Let $P_i \in \mathcal{P}^{\green,\blue}$ be the path whose extension contains $e$. Then the minimal-covering path $P_{c'} \in \mathcal{P}^{\green,\blue}$ of $P_i$ needs to have lower index than $P_c$, i.e. $c' < c$. As $e$ is covered by $P_c$, it is not covered by the minimal-covering path of $P_i$. Hence $e$ will not be added to $D$, which is a contradiction.
			
			First, consider the case $a \neq 1$, i.e. there at least two faces below $P_c$. By Definition~\ref{def:opp}\ref{def:opp_at_least_one_neighbor}, every $f_j$, $j \in \{1, \ldots, a\}$ has at most two edges of extensions of paths in $\{P'_1, \ldots, P'_a\}$ on the boundary. For $j \in \{2, \ldots, a-1\}$ we add at most one of those edges to $D$ and hence $\deg_{\neg H^*+D^*}(f_j^*) \leq 4$ for every $j \in \{2, \ldots, a-1\}$ (see Figure~\ref{fig_illustration_for_proof}).
			
			So let us consider $f_1^*$. Let $P_1' = \{p_1, \ldots, p_b\}$. In Case~1 we add at most one edge of the boundary of $f_1$ to $D$, hence $\deg_{\neg H^*+D^*}(f_1^*) \leq 4$. In Case~2 we know that $v_0v_1$ is \red-\blue-colored since $P_c$ is an index maximal subpath. So by Corollary~\ref{cor:crossingvertex} $(v_0v_1)^*$ is unidirected \green-colored and outgoing at $f_1^*$ and $\deg_{\neg H^*}(f_1^*) \leq 2$. We add at most two edges of the boundary of $f_1$ to $D$ and hence $\deg_{\neg H^*+D^*}(f_1^*) \leq 4$ (see Figure~\ref{fig_illustration_for_proof} for illustration).
			
			Consider $f_a^*$. Let $P_l' = \{w_1, \ldots, w_t\}$. If $v_kv_{k+1}$ is \red-\green-colored, then, by Corollary~\ref{cor:crossingvertex}, $(v_kv_{k+1})^*$ is unidirected \blue-colored and outgoing at $f_a^*$ and hence $\deg_{\neg H^*}(f_a^*) \leq 2$. We add at most two edges of the boundary of $f_a$ to $D$ and hence $\deg_{\neg H^*+D^*}(f_a^*) \leq 4$. So assume that $v_kv_{k+1}$ is unidirected \green-colored. Then $P_c$ is not an index maximal subpath and we are in Case~1. As we observed above $w_1 \neq v_{k+1}$. Hence, we add at most one edge of the boundary of $f_a$ to $D$ and we have that $\deg_{\neg H^*+D^*}(f_a^*) \leq 4$.
			
			Consider the case $a=1$, i.e. there is exactly one face below $P_c$. If $P_c$ is an index maximal subpath, then, by the same arguments as above, we know that $(v_kv_{k+1})^*$ and $(v_1v_0)^*$ are unidirected and outgoing at $f_1^*$. So $\deg_{\neg H^*}(f_1^*) \leq 1$. We add at most three edges of the boundary of $f_1$ to $D$. Those potential edges are an edge of the extension of $P_c$, an edge of the extension of $P'_1$ and an edge of the extension of $P'_2$. If $P_c$ is not an index maximal subpath, then we can use the same arguments which we used to show that $\deg_{\neg H^*+D^*}(f_a^*) \leq 4$ for $a \neq 1$.
			
			Observe that there are faces that are never below a path of $\mathcal{P}_{cover}$. For those faces there is at most one edge of the boundary in $D$. Thus their dual vertices in $\neg H^*+D^*$ have degree at most 4 (see Figure~\ref{fig_illustration_for_proof}).
			
			The clockwise path from $r_2$ to $r_3$ on the outer face boundary is not an index maximal subpath. So no edge on the counterclockwise path from $r_2$ to $r_3$ on the outer face boundary is in $D$. And the only edge which is in $D$ and on the boundary of the outer face is the outgoing \green-colored edge at $r_1$.
			
			So we showed that $H - D$ is acyclic and $\neg H^*+D^*$ has maximum degree at most 4. We now apply the same arguments to $\neg H^* \cup \{b_1b_2, b_2b_3, b_3b_1\}$ obtaining $D'$. The vertices $b_1$, $b_2$ and $b_3$ are the roots of ${G^\sigma}^*$, see Definition~\ref{def:schnyderdual}\ref{def:schnyderdual3}. The edges $b_1b_2$, $b_2b_3$ and $b_3b_1$ are not in $G^*$ and there is only one edge on the boundary of the outer face of $G$ that is also in $D$. Thus we may disregard $b_1b_2$, $b_2b_3$ and $b_3b_1$ in the following and freely switch from $\neg H^* \cup \{b_1b_2, b_2b_3, b_3b_1\}$ to $\neg H^*$.
			
			As shown above the graphs $\neg H^*-D' +D^*$ and $H -D +D'^*$ have maximum degree at most 4 and by construction $\neg H^*-D' +D^* = \neg (H-D+D'^*)^*$. An edge set $E \subseteq E(G)$ is the edge set of a cycle  in $G$ if and only if the edge set $E^*$ is a minimal cut in $G^*$ \cite[Prop. 4.6.1]{Diestel2012}. So in order to show that $\neg H^*-D' +D^*$ and $H-D+D'^*$ are both trees it suffices to show that they are both acyclic. We show that $\neg H^*-D' +D^*$ is acyclic. As before the same arguments work for $H-D+D'^*$. 
			
			For the sake of contradiction, assume that there is a cycle $C$ in $\neg H^*-D' +D^*$. By construction, every cycle in $\neg H^*$ has at least one edge which is also in $D'$. Hence $C$ has at least one edge of $D^*$. Remember that every edge of $D$ is on a cycle of $H$. So by~\cite[Prop. 4.6.1]{Diestel2012} every edge in $D^*$ joins two vertices of two different connected components of $\neg H^*$. 
			
			For a connected component $K$ of $\neg H^*$ let $E_K \subseteq E(G^*)$ be the the minimal cut separating $K$ and $G^*-K$. Let $C_K$ be the cycle of $G$ with $E(C_K) = E_K^*$ and let $P^{C_K} = P_i \in \mathcal{P}^{\green,\blue}$ be the index maximal subpath of $C_K$. Choose $K$ such that $P^{C_K} = P_i$ has smallest index. Since $C$ is a cycle there are two edges $e,e' \in E_K$ that are also in $C$.
			
			Remember that for each index maximal subpath in $\mathcal{P}_{max}$ we pick exactly one edge of the extension and add it to $D$. So either $e^*$ or $e'^*$ is not in the extension of the index maximal subpath $P^{C_K}$. Assume w.l.o.g. that $e^*$ is not in the extension of $P^{C_K}$. Let $P' = P_j \in \mathcal{P}^{\green,\blue}$, $j \in \{1,\dots, s\}$ be the index maximal path  such that $e^*$ is in the extension of $P'$. Since $P^{C_K}$ is the index maximal subpath of $C_K$ we have that $j < i$. So there exists a connected component $K'$ of $\neg H^*$ such that $K'$ and $C$ have a vertex in common and $P'$ is the index maximal subpath of the cycle $C_{K'}$ with $(E(C_{K'}))^*$ being the minimal cut separating $K'$ and $G^*-K'$. This contradicts the definition of $K$. So $\neg H^*-D' +D^*$ and $H-D+D'^*$ are our desired trees.	
		\end{proof}
		
		\begin{corollary}
			Every 3-connected planar graph $G$ contains a 4-tree $T$ whose co-tree $\neg T^*$ is also a 4-tree.
		\end{corollary}
		
		\begin{corollary}
			$r_1$ is a leaf in $H-D+D'^*$ and all edges on the outer face of $G$ except for the outgoing \green-colored edge at $r_1$ are in $H-D+D'^*$. We have $\deg_{H-D+D'^*}(r_3) =2$ and $\deg_{H-D+D'^*}(r_2) \leq 3$. Also the dual vertex of the outer face of $G$ is a leaf in $\neg H^*-D' +D^*$. 
		\end{corollary}
		\begin{proof}
			The proof of Theorem~\ref{thm_max_deg_4} yields that all edges on the outer face of $G$ except for the outgoing \green-colored edge at $r_1$ are in $H-D+D'^*$. In ${G^\sigma}^*$ the path $P_1 \in \mathcal{P}^{\green,\blue}$ is given by the duals of the unidirected incoming \red-colored edges at $r_1$. See Figure~\ref{fig:completion} for illustration. Since the outgoing \green-colored and the outgoing \blue-colored edge at $r_1$ are bidirected, $P_1$ is not an index maximal subpath and hence none of the duals of the unidirected incoming \red-colored edges at $r_1$ is added to $D'$. So $r_1$ is a leaf in $H-D+D'^*$. 
			
			The dual edges of the incoming unidirected edges at $r_2$ and $r_3$ are all covered by the last singleton $b_1$ of $\mathcal{P}^{\green,\blue}$ of $\neg H^* \cup \{b_1b_2, b_2b_3, b_3b_1\}$. See Figure~\ref{fig:completion} for illustration. Let $e_2$ be the dual of the clockwise first unidirected \green-colored incoming edge at $r_2$ and $e_3$ be the dual of the counterclockwise first unidirected \blue-colored incoming edge at $r_3$. Let $I_i$ be the set of the duals of the unidirected $i$-colored incoming edges at $r_i$, $i = \green,\blue$. For $e \in I_i$, $i = \green,\blue$ let $P_{e} \in \mathcal{P}^{\green,\blue}$ be the path such that $e$ belongs to the extension of $P_{e}$. Observe that for all edges $e \in (I_2 \setminus \{e_2\}) \cup (I_3 \setminus \{e_3\})$ $b_1$ is not the minimal-covering path of $P_e$. So those edges are not added to $D'$. On the other hand $b_1$ might be the minimal-covering path of $P_{e_2}$ and/or $P_{e_3}$. Since we added $b_1b_2$ to $D'$ we do not add $e_3$ to $D'$ but might do so for $e_2$. Compare Case~2.2 in the proof of Theorem~\ref{thm_max_deg_4}. Hence $\deg_{H-D+D'^*}(r_3) =2$ and $\deg_{H-D+D'^*}(r_2) \leq 3$.
			
			Since the outgoing \green-colored edge a $r_1$ is the only edge on the boundary of the outer face $f$ which is not in $H-D+D'^*$ we know that the vertex $f^*$ is a leaf in $\neg H^*-D' +D^*$.
		\end{proof}
		
		\begin{remark}\label{rem_int-3-conn}
			There exist internally 3-connected graphs $G_k$ such that every spanning tree of the dual graph has maximum degree at least $\lceil k/2 \rceil$.
		\end{remark}
		\begin{proof}
			In order to define $G_k$ take a cycle $C_k$ on $k$ vertices with fixed embedding. Let $w_0, \dots , w_{k-1}$ be the vertices of the cycle in clockwise order. For every $i =0,\dots,k-1$ add a vertex $p_i$ in the outer face and add edges $p_iw_i$ and $p_iw_{i+1}$ such that the resulting graph $G_k$ is still plane, here indices are modulo $k$. For an illustration see Figure~\ref{fig_Gk}. $G_k$ is internally 3-connected. In the dual of $G_k$ there are multi-edges, i.e. there are vertex pairs that are joined by more than one edge. The graph in which all those vertex pairs are only joined by one edge is the complete bipartite graph $K_{2,k}$. A spanning tree of $K_{2,k}$ has maximum degree at least $\lceil k/2 \rceil$ by pigeonhole principle. 
		\end{proof}
		
		\begin{figure}[!htb]
			\centering
			\begin{tikzpicture}[scale=0.7]
				\node[] (labelnode1) at (0,0) {$f_1$};
				\node[] (labelnode2) at (3,3) {$f_2$};
				
				\foreach \x in {0,...,10}{
					\node[rnode] (\x) at (360/11 * \x : 3) {};
					\node[rnode] (o\x) at (360/11 * \x + 360/22  : 3.5) {};
				}
				
				\foreach \x in {0,...,10}{
					\pgfmathtruncatemacro{\y}{mod(\x+1,11)}
					\draw[nonarrow] (\x) to (\y);
					\draw[nonarrow] (\x) to (o\x);
					\draw[nonarrow] (o\x) to (\y);
				}
				
			\end{tikzpicture}
			\caption{The graph $G_{11}$ of Remark~\ref{rem_int-3-conn}. In a spanning tree of the dual graph $f_1^*$ or $f_2^*$ has degree at least 6.}
			\label{fig_Gk}
		\end{figure}
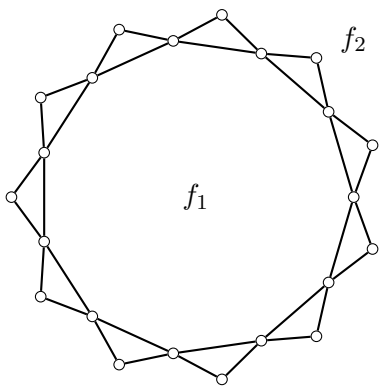
		
		\section{Computational Aspects}\label{sec:computational}
		It is known since 2005 that a minimal Schnyder wood can be computed in linear time $O(n)$~\cite[Section~4, p.\ 60]{Fusy2005}, where $n$ is the number of vertices. Using leftist canonical orderings, an ordered path partition that is compatible to a minimal Schnyder wood can be computed in linear time $O(n)$~\cite[Theorem~7]{Badent2011}. For every path of the compatible ordered path partition, we can detect in time $O(n)$ whether it is the index maximal path of a cycle of the candidate graph $H$. Since the case distinction in our proof, which edges are added to $D$ can be made in linear time for every covering path, we obtain an algorithm with running time $O(n^2)$ to compute a 4-tree whose co-tree is also a 4-tree.
		
		\bibliographystyle{abbrv}
		\bibliography{paper}
\end{document}